\documentclass[12pt,reqno] {amsart}

 \usepackage{amsmath, wasysym}
 \usepackage{bbm, dsfont}
 \usepackage{amssymb}
\usepackage{graphicx}%
\usepackage{hyperref}
\usepackage{cleveref}
\usepackage{comment}
 \usepackage{mathrsfs}
 \usepackage{float}
 \usepackage{tikz, tikz-cd}
  \usetikzlibrary{arrows.meta, knots, calc}
 \usepackage{bm}
 \usepackage{mathtools} 
\usepackage{extarrows} 
\usepackage{bbm}

 \usepackage{enumerate}
  \usepackage{mathdots}
\usetikzlibrary{decorations.markings}
\usetikzlibrary{decorations.pathreplacing}

\theoremstyle{plain}
\newtheorem{theorem}{Theorem}[section]

\newtheorem{lemma}[theorem]{Lemma}
\newtheorem{proposition}[theorem]{Proposition}
\newtheorem*{ac}{Acknowledgements}

\newtheorem{remark}[theorem]{Remark}

\newtheorem{example}[theorem]{Example}

\newcommand{\bC}{\mathbb{C}}

\newcommand{\bR}{\mathbb{R}}
\newcommand{\cL}{\mathcal{L}}

\newcommand{\cN}{\mathcal{N}}

\newcommand{\cG}{\mathcal{G}}

\newcommand{\bE}{\mathbb{E}}
\newcommand{\cI}{\mathcal{I}}
\newcommand{\sP}{\mathscr{P}}
\newcommand{\Tr}{\mathrm{Tr}}
\newcommand{\Id}{\mathrm{Id}}

\newcommand{\bfM}{\mathbf{M}}

\newcommand{\KMS}{\mathrm{KMS}}

\newcommand{\Ent}{\mathrm{Ent}}
\newcommand{\diag}{\mathrm{diag}}

\allowdisplaybreaks[2]
\renewcommand{\Re}{\operatorname{Re}}

\newcommand{\bfW}{\mathbf{W}}
\newcommand{\bfO}{\mathbf{O}}
\newcommand{\bfR}{\mathbf{R}}
\newcommand{\cl}{\mathrm{class}}

\newcommand{\rev}{\mathrm{rev}}

\title[Quantum Reverse Hypercontractivity]{Quantum Reverse Hypercontractivity for KMS-Symmetric Quantum Markov Semigroups}
\author{JINSONG WU}
\date{}

\begin{document}

\begin{abstract}
In this paper, we construct primitive KMS-symmetric quantum Markov semigroups for which the $p$-logarithmic Sobolev constant vanishes at $p=1/2$ and a correlated graph-based example for which the $p$-logarithmic Sobolev constant tends to $-\infty$ as $p\to 0^+$.
These examples obstruct a positive uniform reverse-hypercontractivity rate on $(0,1)$. 
By contrast, we show that every primitive KMS-symmetric quantum Markov semigroup satisfies quantum reverse hypercontractivity for $1/2< p\leq q < 1$ with explicit sufficient times.
For each fixed $\delta>0$, a sufficiently strong dephasing perturbation satisfies a uniform positive logarithmic Sobolev inequality on $[\delta, 1]$ and hence quantum reverse hypercontractivity for $\delta< p\leq q<1$ with fixed $0<\delta<1/2$. 
Furthermore, we establish reverse hypercontractivity for graph-based KMS-symmetric quantum Markov semigroups for $\delta< p\leq q< 1$ with fixed $0<\delta<1$ for connected graphs when $n \geq 3$.
\end{abstract}

\maketitle

\section{Introduction}

Quantum Markov semigroups describe the evolution of open quantum systems under the Markovian approximation \cite{Dav74} in the Heisenberg picture. 
Hypercontractivity initiated by Nelson \cite{Nel66} is a fundamental property in the theory of classical Markov semigroups. 
In 1975, Gross\cite{Gross75} proved the equivalence of the hypercontractivity and logarithmic Sobolev inequality.
Mossel, Oleszkiewicz, and Sen \cite{MOS13} studied classical reverse hypercontractivity and its applications. 
Noncommutative hypercontractivity was established by Olkiewicz and Zegarlinski \cite{OlkZeg99}, who proved the equivalence between the logarithmic Sobolev inequality and hypercontractivity.
Later, King \cite{King14} proved hypercontractivity for a family of tracially symmetric quantum Markov semigroups, where the tracial symmetry means that the semigroup is self-adjoint with respect to the trace.
The quantum Markov semigroup is GNS-symmetric (or KMS-symmetric) if the semigroup is self-adjoint with respect to the GNS-inner product (or KMS-inner product respectively). (See \cite{CarMaa17,GaoRou22,FagUma10, VerWir23, Wir26} and the references therein.)
In \cite{GaoWang26, WanWu26}, it is shown that any primitive KMS-symmetric quantum Markov semigroup has  hypercontractivity.  

Another direction in the study of hypercontractivity is reverse hypercontractivity, which was initiated by Cubitt, Kastoryano, Montanaro, and Temme \cite{CKMT15}. 
Beigi, Datta, and Rouzé\cite{BDR20} proved quantum reverse hypercontractivity for GNS-symmetric quantum Markov semigroups via the quantum Stroock–Varopoulos inequality. 
Their work establishes the equivalence between quantum reverse hypercontractivity and the logarithmic Sobolev inequality for $0<p<1$. 
However, their method cannot be directly extended to more general settings, such as KMS-symmetric quantum Markov semigroups. 
Thus, quantum reverse hypercontractivity for KMS-symmetric quantum Markov semigroups remains unknown.

In this paper, we show that quantum reverse hypercontractivity does not always hold for KMS-symmetric quantum Markov semigroups. In fact, the $p$-contractivity of KMS-symmetric quantum Markov semigroups fails for $0<p<1/2$. 
We find that the $p$-logarithmic Sobolev inequality fails at $p=1/2$. 
Moreover, we show that the $p$-logarithmic Sobolev inequality constant tends to $-\infty$ as $p\to 0^+$. 
Hence, it is important to find a proper characterization of quantum reverse hypercontractivity for KMS-symmetric quantum Markov semigroups. We show that quantum reverse hypercontractivity holds for $1/2<p\leq q<1$. 
Moreover, any primitive KMS-symmetric quantum Markov semigroup perturbed by a dephasing generator satisfies quantum reverse hypercontractivity for $\delta\leq  p\leq q<1$ when $0<\delta<1$ is fixed. 
We also obtain quantum reverse hypercontractivity for primitive KMS-symmetric quantum Markov semigroups arising from undirected finite graphs.

The paper is organized as follows.
In Section 2, we recall the basic notations for noncommutative $L^p$ spaces on matrix algebras and the quantum Markov semigroups.
In Section 3, we present a number of criteria that a primitive KMS-symmetric quantum Markov semigroup fail to satisfy the usual quantum reverse hypercontractivity.
Meanwhile, we obtain examples for each criterion.
In Section 4, we obtain some sufficient conditions for the quantum reverse hypercontractivity and obtain quantum reverse hypercontractivity for $1/2<p\leq q\leq 1$.
In Section 5, we show that any primitive KMS-symmetric quantum Markov semigroup perturbed by a dephasing generator has quantum reverse hypercontractivity for $\delta<p\leq q<1$ with $0<\delta<1$ fixed.
In Section 6, we obtain a quantum reverse hypercontractivity for graph-based KMS-symmetric quantum Markov semigroups.

	\begin{ac}
		The paper is based on the numerous communications with ChatGPT Pro 6. 
		J. Wu was supported by grants from Beijing Institute of Mathematical Sciences and Applications.
		J.~Wu was supported by NSFC (Grant no. 12371124). 
	\end{ac}

\section{Preliminaries}

In this section, we shall recall basic facts about $L^p$-spaces for matrix algebras and the primitive KMS-symmetric quantum Markov semigroups.

Suppose that $X\in M_n(\bC)$ and $\rho$ is a faithful state on $M_n(\bC)$ with density $D_\rho$.
Suppose that $s_1, \ldots, s_n$ are eigenvalues of $D_\rho$ and $\displaystyle s_-=\min_{1\leq j \leq n} s_j$, $\displaystyle s_+=\max_\{1\leq j \leq n\} s_j$.
The $p$-functional of $X>0$ for $0\neq p \in \bR$ is defined by
\begin{align*}
\|X\|_{p, \rho} =\left| \Tr\left(\left|D_\rho^{\frac{1}{2p}} X D_\rho^{\frac{1}{2p}} \right|^p\right) \right|^{1/p}.
\end{align*}
Note that $\|\cdot\|_{p, \rho}$ is a norm for $1\leq p <\infty$.
The $p$-functional $\|\cdot\|_{p, \rho}$ satisfies that for $X>0$, 
\begin{align}\label{eq:normeq}
    \|X\|_{-p,\rho}^{-p}=\|X^{-1}\|_{p,\rho}^{p}.
\end{align}
The reverse H\"{o}lder inequality (see \cite[Lemma 1]{BDR20}) is for any $p<1$, $p\neq 0$,$X, Y>0$, 
\begin{align}\label{eq:revholder}
    \langle X, Y \rangle _{\rho} \geq \|X\|_{p, \rho} \|Y\|_{\frac{p}{p-1}, \rho},
\end{align}
where $\langle \cdot, \cdot \rangle_\rho$ is the KMS-inner product, i.e. $\langle X, Y\rangle_{\rho}=\Tr(X^*D_\rho^{1/2} Y D_\rho^{1/2})$ for all $X, Y\in M_n(\bC)$.
The noncommutative martingale convexity inequality of Ricard and Xu \cite[Theorem~1]{RicXu16} to the $\rho$-preserving conditional expectation $X\mapsto\rho(X)I$, we obtain, for any $X\in M_n(\bC)$ and $2\leq p<\infty$, 
	\begin{align}\label{eq:mart}
		\|X\|_{p, \rho}^2 \leq |\rho(X)|^2+(p-1) \|X-\rho(X)I\|_{p, \rho}^2,
	\end{align}
and for $1\leq p<2$,
	\begin{align}\label{eq:mart1}
		\|X\|_{p, \rho}^2 \geq |\rho(X)|^2+(p-1) \|X-\rho(X)I\|_{p, \rho}^2.
	\end{align}
% The hypercontractivity states that  there exists $\beta_{rev}>0$ such that 
% \begin{align*}
% \|\Phi_t(X)\|_{q, \rho} \geq \|X\|_{p, \rho}, \quad 0 < q \leq p <1, \quad t\geq \frac{1}{\beta_{rev}}\log \frac{1-q}{1-p}
% \end{align*}
For any $0<X\in M_n(\bC)$ and $p, q\neq 0$, the weighted power map $I_{p,q}$ is
\begin{align*}
    I_{p,q}(X) = D_{\rho}^{-\frac{1}{2p}}\left( D_{\rho}^{\frac{1}{2q}} X D_{\rho}^{\frac{1}{2q}}\right)^{q/p} D_{\rho}^{-\frac{1}{2p}}.
\end{align*}
We have that $\|I_{q, p}(X)\|_{q, \rho}^q =\|X\|_{p, \rho}^p$ for all $p, q\neq 0$ and $I_{q, r} I_{r, p}=I_{q, p}$ for all $p, r, q\neq 0$.

The noncommutative Jensen inequality implies that for any unital completely positive map $\Phi$ and positive definite matrix $D$,
\begin{align}\label{eq:jensen}
   \Phi(D^{-1})\geq \Phi(D)^{-1}. 
\end{align}

Next, we shall recall the entropy $\Ent_{p, \rho}$ and its properties.
For $-\infty< p< \infty$ with $p\neq 0$, let
\begin{equation}\label{eq:entropy}
\begin{aligned}
    \Ent_{p, \rho}(X)=&  \Tr\left(\left|D_\rho^{\frac{1}{2p}}XD_\rho^{\frac{1}{2p}}\right|^p \left(\log \left| D_\rho^{\frac{1}{2p}}XD_\rho^{\frac{1}{2p}}\right|^p -\log D_\rho\right)\right) \\
    & - \|X\|_{p, \rho}^p\log \|X\|_{p, \rho}^p.
\end{aligned}
\end{equation}
We have that $\displaystyle \frac{d}{dp}\log \|X\|_{p,\rho}=\frac{\Ent_{p, \rho}(X)}{p^2\|X\|_{p, \rho}^p}$.
By a direct reformulation, we have that for $X>0$,
 \begin{align}\label{eq:entropy2}
     \Ent_{1, \rho}(I_{1, p}(X))=\Ent_{p, \rho}(X).
 \end{align}
By Equation \eqref{eq:normeq}, we have that 
\begin{align}\label{eq:entropyeq}
    \Ent_{-p, \rho}(X)=\Ent_{p, \rho}(X^{-1}).
\end{align}
For $X\neq 0$, we have that 
\begin{align*}
    \Ent_{p, \rho}(X)= \|X\|_{p, \rho}^p  H \left( \left.\frac{|D_\rho^{\frac{1}{2p}}XD_\rho^{\frac{1}{2p}}|^p}{\|X\|_{p, \rho}^p} \right\| D_\rho \right).
\end{align*}
where $H(X\|D)=\Tr(X(\log X-\log D))$ is the relative entropy.

Suppose that $\{\Phi_t\}_{t\geq 0}$ is a quantum Markov semigroup.
We shall write $\Phi_t=e^{-t\cL}$, where $\cL$ is the generator of the semigroup.
The GKSL representation \cite{GKS76,Lin76} of the generator is as follows:
\begin{align}\label{eq:gksl}
\cL(X)=i[H, X] +\sum_{j=1}^m \frac{1}{2}\{V_j^*V_j, X\}-\sum_{j=1}^m V_j^*XV_j,
\end{align}
where $H\in M_n(\bC)$ is a Hermitian matrix and $V_j\in M_n(\bC)$ for $j=1, \ldots,m $.
Suppose that $X>0$.
We have that
\begin{align}\label{eq:gksl2}
  X\cL(X^{-1})X +\cL(X) =-\sum_{j=1}^m [V_j , X]^* X^{-1} [V_j , X]. 
\end{align}

A quantum Markov semigroup $\{\Phi_t\}_{t\geq 0}$ is called
	primitive if it admits a faithful invariant state and
	$\ker(\cL)=\bC I$. Equivalently, it has a unique invariant state
	$\rho$, which is faithful and satisfies $\rho\Phi_t=\rho$
	for all $t\geq 0$.
	
	A quantum Markov semigroup $\{\Phi_t\}_{t\geq 0}$ is KMS-symmetric with respect to a faithful state $\rho$ if for all $t\geq 0$ and $X,Y\in M_n(\bC)$,
	\begin{align}
		\Tr(X^* D_{\rho}^{1/2}\Phi_t(Y)D_\rho^{1/2})=\Tr(\Phi_t(X)^* D_{\rho}^{1/2}YD_\rho^{1/2}),
	\end{align}
	where $D_\rho$ is the density matrix of $\rho$, i.e. $\rho(\cdot)=\Tr(D_\rho \cdot)$.
	This implies that $\{\Phi_t\}_{t\geq 0}$ is in equilibrium with respect to $\rho$.
	For a quantum Markov semigroup admitting a faithful invariant state,
	we denote by $\bE_\Phi$ its ergodic projection onto $\ker(\cL)$.
The Fisher information $\cI_{\cL}$ of the semigroup is
	\begin{align*}
		\cI_{\cL}(D)=\Tr(\cL^*(D)(\log D-\log \bE_{\Phi}^*(D))), 
	\end{align*}
	where $D$ is a faithful density matrix in $M_n(\bC)$ and $\cL^*$ is the adjoint of $\cL$ with respect to $\Tr$.
	When $\{\Phi_t\}_{t\geq 0}$ is primitive, we have that $\cI_{\cL}(D)=\Tr(\cL^*(D)(\log D-\log D_\rho))$.

For $X>0$, the $p$-Dirichlet form $\mathcal{E}_{p, \cL}$ with respect to $\cL$ is 
\begin{align}\label{eq:pvariance}
    \mathcal{E}_{p, \cL}(X)=\frac{p^2}{4(p-1)}\left\langle I_{\frac{p}{p-1}, p}(X), \cL(X)\right\rangle_{\rho}, \quad p\neq 1.
\end{align}
When $p=1$, 
\begin{align}\label{eq:dirfish}
   \mathcal{E}_{1, \cL}(X)=\frac{1}{4}\cI_{\cL}(D_\rho^{1/2} X D_\rho^{1/2}).
\end{align}
For $p\neq 1$ and $X>0$, we have that 
\begin{align*}
    \left.\frac{d}{dt} \log \|\Phi_t(X)\|_{p, \rho}\right|_{t=0} =\frac{4(1-p)}{p^2\|X\|_{p, \rho}^p}\mathcal{E}_{p, \cL}(X).
\end{align*}

The logarithmic Sobolev inequality constant $\alpha_p(\cL)$ for $0<p<\infty$ is defined to be
\begin{align*}
    \alpha_p(\cL) =\inf_{0<X\notin \bC I} \frac{\mathcal{E}_{p, \cL}(X)}{\Ent_{p,\rho}(X)}.
\end{align*}
The modified logarithmic Sobolev inequality states that 
\begin{align*}
    \cI_{\cL}(D)\geq \alpha_{MLSI} H(D\| D_\rho),
\end{align*}
where $\alpha_{MLSI}$ is the optimal constant.
This shows that 
\begin{align}\label{eq:mlsi1}
\mathcal{E}_{1, \cL}(X) \geq \frac{\alpha_{MLSI}}{4}\Ent_{1, \rho}(X).
\end{align}
By Proposition 10 in \cite{BDR20}, we have that
\begin{align}\label{eq:lsiconjugate}
\alpha_p(\cL)=\alpha_{\frac{p}{p-1}}(\cL), \quad p\neq 0,1.
\end{align}

Now we recall some facts about the hypercontractivity.
Note that $\|\Phi_t(X)\|_{p, \rho} \geq \|X\|_{p, \rho}$ for all $X>0$ (which is reverse $p$-contractive) if and only if $\displaystyle \left. \frac{d}{dt}\log \| \Phi_t(X)\|_{p, \rho}\right|_{t=0}\geq 0$ for every $X>0$.
In \cite{BDR20}, they showed that any primitive quantum Markov semigroup is (reverse) $p$-contractivity for $p\in (-\infty, -1]\cup [1/2, \infty)$.
Let 
\begin{align*}
    \beta_{\rev}=\inf_{0< p\leq 1} \alpha_p(\cL).
\end{align*}
Theorem 11 in \cite{BDR20} states that if $\beta_{\rev}>0$, then 
\begin{align}\label{eq:revhyper}
    \|\Phi_t(X)\|_{p, \rho} \geq \|X\|_{q, \rho}, \quad 0<p\leq q< 1,
\end{align}
whenever $\displaystyle t\geq \frac{1}{4\beta_{\rev}}\log \frac{1-p}{1-q}$.

\section{Failure of Quantum Reverse Hypercontractivity}

In this section, we show that $p$-logarithmic Sobolev inequality could fail for primitive KMS-symmetric quantum Markov semigroups when $\displaystyle 0<p\leq \frac{1}{2}$.
Moreover, the corresponding reverse $p$-contractivity fails for $0< p <1/2$.
This leads the failure of quantum reverse hypercontractivity for primitive KMS-symmetric quantum Markov semigroups.

\begin{proposition}\label{prop:lsiconstant1}
Suppose that $\{\Phi_t\}_{t\geq 0}$ is a primitive KMS-symmetric quantum Markov semigroup.
We have that for $p >1$,
\begin{align}
\alpha_{\frac{p}{p+1}}(\cL)\geq \frac{p-1}{p+1}\alpha_p(\cL).
\end{align}
\end{proposition}
\begin{proof}
By Appendix A in \cite{BDR20}, we have that 
\begin{align}\label{eq:pnorm}
  \|\Phi_t(D)\|_{-p, \rho}  \geq \|\Phi_t(D^{-1})\|_{p, \rho}^{-1}.
\end{align}
By taking the differentiation $\log \|\Phi_t(D)\|_{p, \rho}$ with respect to $t$ at $t=0$, we obtain that 
\begin{align*}
   \left. \frac{d}{dt}\right|_{t=0} \log \|\Phi_t(D)\|_{p, \rho}
   =\frac{4(1-p)}{p^2\|D
   \|_{p, \rho}^p}\mathcal{E}_{p, \cL}(D).
\end{align*}
Note that inequality \eqref{eq:pnorm} is an equality when $t=0$.
Hence
\begin{align}
    \mathcal{E}_{-p, \rho}(D)\geq \frac{p-1}{p+1}\mathcal{E}_{p, \cL}(D^{-1}).
\end{align}
By Equation \eqref{eq:entropyeq}, we have that
\begin{align}\label{eq:mlsi32}
    \alpha_{-p}(\cL) \geq \frac{p-1}{p+1} \alpha_p(\cL).
\end{align}
By Equation \eqref{eq:lsiconjugate}, we have that $\alpha_{\frac{p}{p+1}}(\cL) =\alpha_{-p}(\cL)$.
Combining with Equation \eqref{eq:mlsi32}, we see that the proposition is true.
\end{proof}

\begin{proposition}\label{prop:lsiconstant2}
Suppose that $\{\Phi_t\}_{t\geq 0}$ is a primitive KMS-symmetric quantum Markov semigroup.
We have that for $p >1$,
\begin{align*}
    \alpha_{\frac{p}{p+1}}(\cL)\geq \frac{1}{2} \frac{p-1}{p+1} \frac{\lambda_{\KMS}}{\log s_-^{-1}+2},
\end{align*}
where $\displaystyle \lambda_{\KMS}=\inf_{X\neq 0, \rho(X)=0} \frac{\langle X, \cL(X)\rangle _\rho}{\langle X, X \rangle_\rho}$.
\end{proposition}
\begin{proof}
By Proposition \ref{prop:lsiconstant1} and logarithmic Sobolv inequality in \cite{WanWu26, GaoWang26}, we see that the proposition is true.    
\end{proof}

\begin{remark}
Proposition \ref{prop:lsiconstant2} implies that $\alpha_{p}(\cL)>0$ for $\displaystyle \frac{1}{2}< p <1$.   
Combining with Theorem 7.3 in \cite{WanWu26}, we have that the $p$-logarithmic Sobolev inequality holds for primitive KMS-symmetric quantum Markov semigroup whenever $\displaystyle p>\frac{1}{2}$.
\end{remark}

\begin{remark}
    By Theorem 3.6 in \cite{WanWu26}, we have that for $\frac{1}{2}<p<1$,
    \begin{align*}
    \alpha_p(\cL) \geq \frac{(2p-1)h_\rho(\cL)^2s_-^2}{4s_+^2\|\cL\|_{\KMS,\rho}(\log s_-^{-1}+2)},
    \end{align*}
    where $\displaystyle h_\rho(\cL)=\inf_{0, I\neq P \in \sP}\frac{\langle P, \cL(P)\rangle_{\KMS,\rho}}{\min\{\rho(P), 1-\rho(P)\}}$ is the projection Dirichlet conductance defined in \cite{WanWu26}.
\end{remark}

\begin{proposition}\label{prop:halfzero}
    Suppose that $\{\Phi_t\}_{t\geq 0}$ is a primitive KMS-symmetric quantum Markov semigroup with generator \eqref{eq:gksl}.
    If there exists nonscalar positive definite matrix $X_0$ in $M_n(\bC)$ such that $[V_j, X_0]=0$ for all $j$,
    then 
    \begin{align*}
        \alpha_{1/2}(\cL)=0.
    \end{align*}
\end{proposition}
\begin{proof}
First of all, we shall show that for any $X>0$,
\begin{align}\label{eq:halfe}
    \mathcal{E}_{1/2, \cL}(|D_\rho^{1/2} X D_\rho^{-1/2}|^2)
    =\frac{1}{8} \sum_{j=1}^m \Tr\left( D_\rho [V_j, X]^* X^{-1}[V_j, X]\right).
\end{align}
Following the definition and the KMS symmetry, we have that 
\begin{align*}
    \mathcal{E}_{1/2, \cL}(|D_\rho^{1/2} X D_\rho^{-1/2}|^2)
    =& -\frac{1}{8}\Tr\left(D_\rho^{-1/2}X^{-1}D_\rho^{-1/2} D_\rho \cL(|D_\rho^{1/2} X D_\rho^{-1/2}|^2) D_\rho\right) \\
    =& -\frac{1}{8}\Tr\left( D_\rho X \cL(X^{-1})X\right).
\end{align*}
Now applying Equation \eqref{eq:gksl2} and $\cL^*(D_\rho)=0$, we see that 
\begin{align*}
  \mathcal{E}_{1/2, \cL}(|D_\rho^{1/2} X D_\rho^{-1/2}|^2)
  =& \frac{1}{8}\Tr\left( D_\rho  \sum_{j=1}^m [V_j, X]^* X^{-1}[V_j, X]\right) +\frac{1}{8}\Tr(D_\rho \cL(X))\\
  =& \frac{1}{8}\Tr\left( D_\rho  \sum_{j=1}^m [V_j, X]^* X^{-1}[V_j, X]\right).
\end{align*}
This completes the proof of the identity \eqref{eq:halfe}.

Finally, we see that if there exists $X_0> 0$ such that $[V_j, X_0]=0$ for all $j=1, \ldots, m$, then $\mathcal{E}_{1/2, \cL}(|D_\rho^{1/2} X_0 D_\rho^{-1/2}|^2)=0$.
Note that $\Ent_{1/2,\rho}(|D_\rho^{1/2} X_0 D_\rho^{-1/2}|^2)>0$.
We see that $\alpha_{1/2}(\cL)=0$.
\end{proof}

\begin{example}
In the following, we shall find a primitive KMS-symmetric quantum Markov semigroup with $\alpha_{1/2}(\cL)=0$.
Let
\begin{align*}
D_\rho^{1/2}=\frac{1}{\sqrt{94}}\begin{pmatrix}
    2 & 0 & 1 \\ 0 & 8 & 2\\ 1 & 2 & 4
\end{pmatrix}, \quad
D_\rho=\frac{1}{94}\begin{pmatrix}
    5 & 2 & 6 \\ 2 & 68 & 24\\ 6 & 24 & 21
\end{pmatrix}, \\
V= \begin{pmatrix}
    0 & 1 & 0 \\ 4 & 0 & 0\\ 0 & 0 & 2
\end{pmatrix},\quad
H=\frac{i}{37}\begin{pmatrix}
    0 & -9 & 40 \\ 9 & 0 & -10 \\ -40 & 10 & 0
\end{pmatrix}.
\end{align*}
The semigroup generator $\cL$ is given by
\begin{align}
    \cL(X)=-i[H, X]+\frac{1}{2}\{V^*V, X\}-V^*XV.
\end{align}
By a direct computation, we have that 
\begin{align*}
    VD_\rho^{1/2}=D_\rho^{1/2}V^*, \quad \left(\frac{1}{2}V^*V+iH\right)D_\rho^{1/2}= D_\rho^{1/2}\left(\frac{1}{2}V^*V-iH\right).
\end{align*}
This shows that $\cL^*(D_\rho^{1/2}XD_\rho^{1/2})=D_\rho^{1/2}\cL(X)D_\rho^{1/2}$ and the semigroup is KMS-symmetric.
% Note that $H\neq 0$. It is not GNS-symmetric.

Suppose that $X_0\in M_n(\bC)$ and $\cL(X_0)=0$.
Then by the gradient form and faithfulness of $\rho$, we have that $[V, X_0]=[V^*, X_0]=0$.
This indicates that $X_0$ is $\diag(x,x,y)$, where $x, y\in \bC$.
For such $X_0$, $\cL(X_0)=-i[H, X_0]$, hence $[H,X_0]=0$.
It shows that $X_0$ is a multiple of the identity.
Hence the semigroup is primitive KMS-symmetric.

By taking 
\begin{align*}
    X_0=\begin{pmatrix}
    1 & 0 & 0 \\ 0 & 1 & 0 \\ 0 & 0 & 2
\end{pmatrix},
\end{align*}
we see that $[V, X_0]=0$.
By Proposition \ref{prop:halfzero}, we see that $\alpha_{1/2}(\cL)=0$.
Moreover for $\varepsilon>0$ small enough, we have that
\begin{align*}
    \left\|\Phi_{t_\varepsilon}(|D_\rho^{1/2} X_0 D_\rho^{-1/2}|^2)\right\|_{1/2, \rho} < \left\||D_\rho^{1/2} X_0 D_\rho^{-1/2}|^2\right\|_{q_\varepsilon, \rho},
\end{align*}
where $\displaystyle t_\varepsilon=-\frac{1}{4\beta}\log(1-2\varepsilon)$ for some $\beta>0$ and $\displaystyle q_\varepsilon =\frac{1}{2}+\varepsilon$, which contradicts with the proposed reverse hypercontractivity estimate for small $\varepsilon$.
\end{example}

In the following, we shall analyze the logarithmic Sobolev inequality constant at $p=0$.
To proceed this analysis, we shall define three Schur multipliers $\bfW_{p, \rho}, \bfM_{\rho}, \bfR_{p, \rho}$ as follows.
Suppose that $\{E_{jk}\}_{j,k=1}^n$ is the system of matrix units. 
Let
\begin{align*}
\bfW_{p, \rho} (E_{jk}) =w_p(s_j, s_k) E_{jk},\quad 
\bfM_\rho (E_{jk}) =\min\{s_j, s_k\} E_{jk}, \quad
\bfR_{p,\rho} (E_{jk}) =r_p(s_j, s_k) E_{jk},
\end{align*}
 where 
\begin{align*}
w_p(s,t) =\frac{s^{1/p}t^{1/p}}{p-1} \frac{s^{1-1/p}-t^{1-1/p}}{s^{1/p}-t^{1/p}}, \quad s\neq t; \quad 
w_p(s,s)=s ,
\end{align*}
and
\begin{align*}
r_p(s, t) = s^{1/p}t^{1/p}\frac{(\log s-\log t)(s-t)}{p^2 (s^{1/p}-t^{1/p})^2}, \quad 
r_p(s,s)= s
\end{align*}
Note that $\displaystyle \lim_{p\to 0^+} w_p(s_j s_k) E_{jk}= \min\{s_j, s_k\} E_{jk}$ and $\displaystyle \lim_{p\to 0^+} r_p(s_j s_k) E_{jk}= 0$ for $s_j\neq s_k$.
We have that $\displaystyle \lim_{p\to 0^+} \bfW_{p, \rho}=\bfM_{\rho}$.
Moreover, we have that the three Schur multipliers are completely positive maps.

\begin{proposition}\label{prop:mlsibd1}
Suppose that $0\neq X=X^*$ with $\rho(X)=0$. 
Then for any $0<p<1$, we have that 
\begin{align}\label{eq:mlsibd1}
\alpha_p(\cL) \leq \frac{\Re \Tr(X \bfW_{p, \rho} \cL(X))  }{2\Tr(X \bfR_{p, \rho}(X))}.
\end{align}
\end{proposition}
\begin{proof}
By a direct computation, we have that
\begin{align*}
\left. \frac{d}{dt}\log \| \Phi_t(I+\varepsilon X)\|_{p, \rho}\right|_{t=0} =(1-p) \varepsilon^2 \Tr(X\bfW_{p, \rho}(\cL(X))) +\bfO(\varepsilon^3).
\end{align*}
Hence
\begin{align}\label{eq:energe1}
\mathcal{E}_{p, \cL} (I+\varepsilon X) =\frac{p^2\varepsilon^2}{4}  \Re \Tr(X \bfW_{p, \rho} \cL(X)) +\bfO(\varepsilon^3)
\end{align}
Similarly, we have that 
\begin{align}\label{eq:entropyepsilon}
\Ent_{p, \rho}(I+\varepsilon X) =\frac{p^2\varepsilon^2}{2} \Tr(X \bfR_{p, \rho}(X)) +\bfO(\varepsilon^3).
\end{align}
Combining the two equalities, we see that the proposition is true.
\end{proof}

\begin{theorem}\label{thm:0infty}
Suppose that $\{\Phi_t\}_{t\geq 0}$ is a primitive KMS-symmetric quantum Markov semigroup.
Let $\cN$ be the commutant $\{D_\rho\}'$ and $\bE_{\cN}$ the trace-preserving conditional expectation onto $\cN$.
Suppose that there exists $X_0=X_0^*$ with $\bE_{\cN}(X_0)=0$ such that
\begin{align*}
\Re \Tr(X_0 \bfM_\rho \cL(X_0))<0.
\end{align*}
Then \begin{align*}
\inf_{0<p < 1} \alpha_p(\cL) =-\infty.
\end{align*}
\end{theorem}
\begin{proof}
For the denominator of Equation \eqref{eq:mlsibd1}, we have that $\lim_{p\to 0^+}\bfR_{p, \rho}=D_\rho \bE_{cN}(X_0)$ and 
\begin{align*}
\Tr(X_0 \bfR_{p, \rho}(X_0))>0, \quad  \lim_{p\to 0^+}\Tr(X_0 \bfR_{p, \rho}(X_0))=0.
\end{align*}
For the numerator, we have that 
\begin{align*}
\lim_{p\to 0^+} \Tr(X_0 \bfW_{p, \rho} \cL(X_0))= \Tr(X_0 \bfM_{\rho} \cL(X_0)).
\end{align*}
By the assumption, we have that 
\begin{align*}
\lim_{p\to 0^+} \alpha_p(\cL)=-\infty.
\end{align*}
This completes the proof of the theorem.
\end{proof}

\begin{remark}
Note that for GNS-symmetric quantum Markov semigroups, the map $p\mapsto \alpha_p(\cL)$ is non-increasing on $(0,1)$ as proved in \cite{BDR20}.
Hence Theorem \ref{thm:0infty} shows a big difference via the logarithmic Sobolev inequality between the KMS-symmetric quantum Markov semigroups and the GNS-symmetric quantum Markov semigroups.
\end{remark}

\begin{example}\label{example:infty}
We consider a graph-based KMS-symmetric quantum Markov semigroup on $M_4(\bC)$ with nonzero Hamiltonian.
Let
\begin{align*}
D_\rho=\frac{1}{30}\begin{pmatrix}
4 & 0 & 0 & 0 \\ 0 & 1 & 0& 0 \\
0 & 0 & 16& 0 \\ 0 & 0 & 0 & 9
\end{pmatrix}, \quad 
V_1=\begin{pmatrix}
0 & 2 & 0 & 0 \\ 1 & 0 & 0& 0 \\
0 & 0 & 0& 0 \\ 0 & 0 & 0 & 0
\end{pmatrix}
\quad
V_2= \begin{pmatrix}
0 & 0 & 0 & 0 \\ 0 & 0 & 1& 0 \\
0 & 4& 0& 0 \\ 0 & 0 & 0 & 0
\end{pmatrix},
\end{align*}
and 
\begin{align*}
V_3=\begin{pmatrix}
0 & 0 & 0 & 0 \\ 0 & 0 & 0& 0 \\
0 & 0 & 0 & 4 \\ 0 & 0 & 3 & 0
\end{pmatrix}, \quad 
V=\begin{pmatrix}
0 & 6 & 0 & 0 \\ 3 & 0 & 1 & 0 \\
0 & 4 & 0& 4 \\ 0 & 0 & 3 & 0
\end{pmatrix}
\quad
H=-i \begin{pmatrix}
0 & 0 & 1/2 & 0 \\ 0 & 0 & 0 & 4 \\
-1/2 & 0& 0& 0 \\ 0 & -4 & 0 & 0
\end{pmatrix}.
\end{align*}
The semigroup generator $\cL$ is 
\begin{align*}
 \cL (X)=  \frac{1}{2}\{V^*V, X\}-V^*XV +i[H, X] +\frac{1}{1000}\left(\frac{1}{2}\left\{\sum_{j=1}^3V_j^*V_j, X\right\}-\sum_{j=1}^3V_j^*XV_j\right).
\end{align*}
This is a graph-based primitive KMS-symmetric quantum Markov semigroup with nonzero Hamiltonian.
Let 
\begin{align*}
    X=\begin{pmatrix}
        0 & 1 & 0 & -2 \\ 1 & 0 & -2 & 0\\
        0 & -2 & 0 & 2 \\ -2 & 0 & 2 & 0
    \end{pmatrix}
\end{align*}
We have that 
\begin{align*}
   \Tr(X\bfM_\rho(\cL(X))) =-\frac{2551}{30000}<0.
\end{align*}
By Theorem \ref{thm:0infty}, we have that $\displaystyle \lim_{p \to 0^+}\alpha_p(\cL)=-\infty$.
\end{example}

\section{Quantum Reverse Hypercontractivity}

In this section, we shall obtain some positive results regarding the quantum reverse hypercontractivity.
\begin{proposition}\label{prop:rev1}
Suppose that $\{\Phi_t\}_{t\geq 0}$ is a primitive KMS-symmetric quantum Markov semigroup.
If there exists $c>0$ such that  for all $0<p<1$ and $X>0$ in $M_n(\bC)$, 
\begin{align}\label{eq:suffice}
    \mathcal{E}_{p, \cL}(X) \geq c\mathcal{E}_{1, \cL}(I_{1, p}(X)),
\end{align}
then the reverse hypercontractivity \eqref{eq:revhyper} holds whenever $\displaystyle t\geq \frac{1}{c\alpha_{MLSI}}\log \frac{1-p}{1-q}$.
\end{proposition}
\begin{proof}
By the assumption, we have that 
\begin{align*}
 \mathcal{E}_{p, \cL}(X) 
 \geq &  c\mathcal{E}_{1, \cL}(I_{1, p}(X))\\
% = &  \frac{c}{4}\cI_{\cL}(X) \quad \text{Equation \eqref{eq:dirfish}}\\
 \geq &   \frac{c\alpha_{MLSI}}{4} \Ent_{1, \rho}(I_{1, p}(X)) \quad \text{Equation \eqref{eq:mlsi1}}\\
 =& \frac{c\,\alpha_{MLSI}}{4}\Ent_{p, \rho}(X), \quad \text{Equation \eqref{eq:entropy2}}.
\end{align*}
By Theorem 11 in \cite{BDR20}, we see that the proposition is true.
\end{proof}

\begin{remark}
    In \cite{BDR20}, they show that $c=1$ in Proposition \ref{prop:rev1} for GNS-symmetric quantum Markov semigroups.
\end{remark}

In the following, we shall find a primitive KMS-symmetric quantum Markov semigroup satisfying the assumption in Proposition \ref{prop:rev1}.

Suppose that $\rho$ is a faithful state on $M_n(\bC)$ and $D_\rho$ is its density which is not a scalar.
We assume that $D_\rho=\diag(s_1, \ldots, s_n)$.
Let $\displaystyle s_-=\min_{1\leq j \leq n} s_j$ and $\displaystyle s_+=\max_{1\leq j \leq n} s_j$.
Now we define a linear map $\cL$ on $M_n(\bC)$ as follows:
\begin{align}\label{eq:transpose}
    \cL(X)=\gamma(X-\rho(X))+\omega(X-X^{\mathsf{T}}), \quad X\in M_n(\bC),
\end{align}
where $\gamma>0$ and $0<\omega\leq \gamma s_-$.

\begin{lemma}
The semigroup $\{\Phi_t\}_{t\geq 0}$ with generator \eqref{eq:transpose} is a primitive KMS-symmetric quantum Markov semigroup.    
\end{lemma}
\begin{proof}
Let $\Psi(X)=\gamma \rho(X)+\omega X^{\mathsf{T}}$.
Then the Choi matrix is $\gamma D_\rho\otimes I +\omega F$, where $F$ is the flip operator.
When $\omega \leq \gamma s_-$, we see that the Choi matrix is positive semidefinite and then $\Psi$ is completely positive.

Note that $\Psi$ is KMS-symmetric with respect to $D_\rho$ and the fixed point space of $(\gamma+\omega)^{-1}\Psi$ is $\bC$.
By taking Hermitian matrix $X$ such that $[X, D_\rho]\neq 0$, we see that $\Psi$ is not GNS-symmetric.

By Lemma 6.1 in \cite{WanWu26}, we have that the semigroup is a primitive KMS-symmetric quantum Markov semigroup by noting that $\displaystyle \cL(X)= (\gamma+\omega)\left(X-\frac{\Psi(X)}{\gamma+\omega}\right)$.
\end{proof}

\begin{proposition}
Suppose that $\{\Phi_t\}_{t\geq 0}$ is the primitive KMS-symmetric quantum Markov semigroup with generator \eqref{eq:transpose}.
Then the reverse hypercontractivity \eqref{eq:revhyper} holds for $0< p\leq q< 1$ whenever
\begin{align*}
    t\geq \frac{1}{\gamma} \log \frac{1-p}{1-q}.
\end{align*}
\end{proposition}
\begin{proof}
Let 
\begin{align*}
    \bE_\rho(X)=\rho(X) I, \quad \Theta(X)=X^{\mathsf{T}},
\end{align*}
and 
\begin{align*}
    \Phi_t^{(1)}=& e^{\gamma t (\bE_{\rho}-\Id)}=e^{-\gamma t} \Id +(1-e^{-\gamma t} )\bE_\rho, \\
    \Phi_t^{(2)}=& e^{\omega t (\Theta-\Id)} =\frac{1+e^{-2\omega t}}{2} \Id +\frac{1-e^{-2\omega t}}{2}\Theta.
\end{align*}
Then $\Phi_t=\Phi_t^{(1)}\Phi_t^{(2)}$.
We see that $\Phi_t^{(2)}$ is a unital positive map and it is not completely positive in general.
By a direct computation, we have that $\|\Phi_t^{(2)}(X)\|_{p, \rho} \geq \|X\|_{p, \rho}$ and 
\begin{align*}
    \| \Phi_t^{(1)}(X)\|_{p, \rho}\geq \|X\|_{q, \rho}, \quad 0< p\leq q \leq 1,
\end{align*}
whenever $\displaystyle t \geq \frac{1}{\gamma} \log \frac{1-p}{1-q}$.
This implies that 
\begin{align*}
    \| \Phi_t(X)\|_{p, \rho}\geq \|X\|_{q, \rho}, \quad 0< p\leq q \leq 1,
\end{align*}
whenever $\displaystyle t \geq \frac{1}{\gamma} \log \frac{1-p}{1-q}$.
This completes the proof of the proposition.
\end{proof}

\begin{remark}
There exists a primitive KMS-symmetric quantum Markov semigroup such that the assumption \eqref{eq:suffice} could fail for all $c>0$.     
Now we focus on the graph-based KMS-symmetric quantum Markov semigroup for $n=2$.
Let
\begin{align*}
D_\rho=\begin{pmatrix}
    1/5 & 0 \\ 0 & 4/5
\end{pmatrix}, \quad 
V=\frac{1}{\sqrt{5}}\begin{pmatrix}
    0 & 1 \\ 2 & 0
\end{pmatrix}
\end{align*}
and the generator $\cL$ to be 
\begin{align*}
 \cL(X)= \frac{1}{2}  \{V^*V, X\}-V^*XV.
\end{align*}
Let 
\begin{align*}
    X_{\varepsilon} =\frac{1}{3}\begin{pmatrix}
        2 & \sqrt{2} \\
        \sqrt{2} & 1
    \end{pmatrix}+\varepsilon I.
\end{align*}
Then 
\begin{align*}
    \mathcal{E}_{1/2, \cL}(X_\varepsilon) =\frac{1-\sqrt{\varepsilon/(1+\varepsilon)}}{50\sqrt{6+17\varepsilon +8\sqrt{\varepsilon(1+\varepsilon)}}} \to \frac{1}{50\sqrt{6}} \quad ( \varepsilon \to 0),
\end{align*}
and 
\begin{align*}
    \mathcal{E}_{1, \cL}(I_{1, 1/2}(X_\varepsilon))
    =\frac{1}{4} \Tr(\cL^*(Y_\varepsilon)(\log Y_\varepsilon-\log D_\rho)
    =\frac{\sqrt{6}}{450} \log \varepsilon^{-1} +O(1),
\end{align*}
where $Y_\varepsilon=  (D_\rho X_\varepsilon D_\rho)^{1/2}  $.
Now we see that 
\begin{align*}
    \frac{\mathcal{E}_{1/2, \cL}(X_\varepsilon) }{\mathcal{E}_{1, \cL}(I_{1, 1/2}(X_\varepsilon))} \sim \frac{3}{2\log \varepsilon^{-1}} \to 0 \quad (\varepsilon \to 0).
\end{align*}
This indicates that for such KMS-symmetric quantum Markov semigroup, we can not find $c>0$ with the assumption \eqref{eq:suffice}.
\end{remark}

In the following, we shall obtain the quantum reverse hypercontractivity for primitive KMS-symmetric quantum Markov semigroups by restricting the ranges of $p,q$.

\begin{theorem}\label{thm:revhyper1}
Suppose that $\varepsilon>0$ is fixed and $\{\Phi_t\}_{t\geq 0}$ is a primitive KMS-symmetric quantum Markov semigroup.   
Then
\begin{align}\label{eq:revhyper1}
    \|\Phi_t(X)\|_{p, \rho} \geq \|X\|_{q, \rho}, \quad \frac{1}{2}+\varepsilon \leq p\leq q< 1,
\end{align}
whenever $\displaystyle t\geq \frac{2+\log s_-^{-1}}{4\varepsilon \lambda_{\KMS}}\log \frac{1-p}{1-q}$.
\end{theorem}
\begin{proof}
By Proposition \ref{prop:lsiconstant2}, we have that 
\begin{align*}
    \inf_{\frac{1}{2}+\varepsilon\leq p \leq 1} \alpha_p(\cL) \geq \varepsilon \frac{\lambda_{\KMS}}{\log s_-^{-1}+2}.
\end{align*}
By the proof of Theorem 11 in \cite{BDR20}, we have that the theorem is true.
\end{proof}

\begin{lemma}\label{lem:integration}
Suppose that $\delta\geq 0$ is fixed and $\{\Phi_t\}_{t\geq 0}$ is a primitive KMS-symmetric quantum Markov semigroup.
If $\alpha_p(\cL)>0$ for all $\delta< p <1$, then the reverse hypercontractivity holds for $\delta < p\leq  q< 1$ whenever
\begin{align*}
    t \geq \frac{1}{4}\int_p^q \frac{ds}{(1-s)\alpha_s(\cL)}.
\end{align*}
\end{lemma}
\begin{proof}
 It follows from the proof of Theorem 11 in \cite{BDR20}.   
\end{proof}

Furthermore, by removing the $\varepsilon$, we have the following reverse hypercontractivity.

\begin{theorem}\label{thm:revhyper2}
Suppose that $\{\Phi_t\}_{t\geq 0}$ is a primitive KMS-symmetric quantum Markov semigroup.   
Then
\begin{align}\label{eq:revhyper2}
    \|\Phi_t(X)\|_{p, \rho} \geq \|X\|_{q, \rho}, \quad \frac{1}{2} < p\leq q< 1,
\end{align}
whenever $\displaystyle t\geq \frac{2+\log s_-^{-1}}{2\lambda_{\KMS}}\log \frac{(2q-1)(1-p)}{(2p-1)(1-q)}$.  
\end{theorem}
\begin{proof}
By Proposition \ref{prop:lsiconstant2}, we have that
\begin{align*}
    \alpha_p(\cL) \geq \frac{(2p-1)\lambda_{\KMS}}{2(2+\log s_-^{-1})}, \quad \frac{1}{2} < r < 1.
\end{align*}
By Lemma \ref{lem:integration}, we have that whenever
\begin{align*}
t \geq &  \frac{2+\log s_-^{-1}}{ 2\lambda_{\KMS}}  \int_p^q \frac{ds}{(2s-1)(1-s)} \\
= & \frac{2+\log s_-^{-1}}{ 2\lambda_{\KMS}} \log \frac{(2q-1)(1-p)}{(2p-1)(1-q)},
\end{align*}
the reverse hypercontractivity \eqref{eq:revhyper2} holds.
\end{proof}

\section{Perturbations by Dephasing Generators}
In this section, we shall show that any primitive KMS-symmetric quantum Markov semigroup can be perturbed by a dephasing generators to have a quantum reverse hypercontractivity.

For the GKSL representation of the generator described in \eqref{eq:gksl}, we define
\begin{align*}
    B_{p, j}= D_\rho^{\frac{1}{2p}} V_j^* D_\rho^{-\frac{1}{2p}}, \quad G_p=D_\rho^{\frac{1}{2p}}\left(\frac{1}{2}\sum_{j=1}^m V_j^*V_j+iH\right) D_\rho^{-\frac{1}{2p}}.
\end{align*}
For $X>0$, we let $X_p= D_\rho^{\frac{1}{2p}} X D_\rho^{\frac{1}{2p}}$.
Suppose that $\displaystyle D_\rho=\sum_{j=1}^n s_j P_j$, where $P_j$ are rank-one projections.
Let $\cN$ be the algebra generated by the rank-one projections $P_j$ and $\bE_{\cN}$ the trace-preserving conditional expectation onto $\cN$.
Now we define a dephasing generator $\cL_{\cN}$ as follows:
\begin{align}
\cL_{\cN}(X) = X-\sum_{j=1}^n P_j XP_j=X-\bE_{\cN}(X), \quad X\in M_n(\bC).
\end{align}
We have that $\displaystyle \cL_{\cN}(X)= \frac{1}{2} \sum_{j=1}^n [P_j, [P_j, X]]$.

We consider the generator $\bE_{\cN}\cL\bE_{\cN}$ on the algebra $\cN$, which is a classical Markov generator, denoted by $\cL_{\cl}$, i.e.
\begin{align*}
    \left(\cL_{\cl}(\vec{t})\right)_{j}=\sum_{k\neq j} w_{jk} (t_j-t_k),
\end{align*}
where $\displaystyle w_{jk}=\sum_{\ell=1}^m |(V_{\ell})_{kj}|^2$.
The KMS symmetry of $\cL$ implies that $\cL_{\cl}$ is reversible with respect to $\vec{s}=(s_1, \ldots, s_n)^{\mathsf{T}}$.
Let $\lambda_{\cl}$ be the spectral gap of $\cL_{\cl}$.
We have that $0<\lambda_{\KMS}\leq \lambda_{\cl}$.
Moreover, the classical Markov chain is irreducible.
Let
\begin{align*}
    \alpha_{\cl}=\alpha_1(\cL_{\cl}).
\end{align*}
By Corollary 16 in \cite{BDR20}, we see that $\alpha_p(\cL_{\cl}) \geq \alpha_1(\cL_{\cl})$ for $0<p<1$.
Note that for any $\vec{t}=(t_1, \ldots, t_n)^{\mathsf{T}}$ with $\displaystyle \sum_{j=1}^n s_j t_j=1$ and $t_j\geq 0$, 
We have that 
\begin{align*}
\cI_{\cL_{\cl}}(\vec{t})\geq s_-\lambda_{\cl} \text{Var}_{\vec{s}} (\vec{t}) \geq  s_-\lambda_{\cl}\Ent_{\vec{s}} (\vec{t}).
\end{align*}
This show that $\displaystyle \alpha_{\cl} \geq \frac{s_- \lambda_{\cl}}{4}$.

\begin{lemma}\label{lem:infty}
 For all $X>0$ with $\rho(X)=1$,  we have that $\mathcal{E}_{p, \cL}(X)\geq -C_{\delta, \rho,\cL}$ for some constant $C_{\delta, \rho,\cL}$ depending only on $\delta, \rho$ and $\cL$.
Moreover,
suppose that $\rho(X_r)=1$, $X_r>0$, $p_r$ in $[\delta, 1]$ and $X_r$ convergent to $X_0$ where $X_0$ is a singular diagonal matrix. 
 Then we have that 
   \begin{align*}
       \mathcal{E}_{p_r, \rho} (X_r)\to \infty.
   \end{align*}
\end{lemma}
\begin{proof}
   Suppose that $X_p$ is in the diagonal form with eigenvalues $\mu_k$.
We have that 
% \begin{align*}
%     \frac{4}{p^2} \mathcal{E}_{p, \cL}(X)
% =\sum_{e, k, \ell} |(B_{p, e})_{k\ell}|^2 \mu_\ell^p \frac{(\mu_\ell/\mu_k)^{1-p}-1}{1-p}
% +\frac{\Tr(\sum_e B_{p, e}^* B_{p,e}-\sum_{e}W_e^*W_e)}{1-p}.
% \end{align*}
\begin{align}\label{eq:energyexpansion}
    \frac{4}{p^2} \mathcal{E}_{p, \cL}(X)
=\sum_{j, k, \ell} |(B_{p, j})_{k\ell}|^2 \mu_\ell^p \frac{(\mu_\ell/\mu_k)^{1-p}-1}{1-p}
+\frac{\Tr\left(X_p^p \left(\sum_{j=1}^m B_{p, j}^* B_{p,j}-G_p- G_p^*\right)\right)}{1-p}.
\end{align}
Let
\begin{align*}
    K_p=\frac{  \sum_{j=1}^m B_{p, j}^* B_{p,j}-G_p- G_p^*  }{1-p}.
\end{align*}
By the fact that $\cL^*(D_\rho)=0$ , we see that $\displaystyle \sum_{j=1}^m B_{1, j}^* B_{1, j}=G_1+G_1^*$ and $K_p$ extends continuous to $p=1$.

For the first term on the right hand side, by the fact that $\displaystyle \frac{(\mu_\ell/\mu_k)^{1-p}-1}{1-p}\geq \log \frac{\mu_\ell}{\mu_k}$, we see that 
\begin{align*}
  \mu_\ell^p\log \frac{\mu_\ell}{\mu_k} \geq \left\{\begin{array}{ll} 
   \mu_\ell^p\log s_-\mu_\ell\geq \displaystyle -\frac{s_-^{-p}}{ep} & \mu_\ell \leq \mu_k\\
   0 & \mu_\ell \geq \mu_k
  \end{array} \right..
\end{align*}
Hence there exists a constant $C_{\delta, \rho, \cL}$ such that $\mathcal{E}_{\rho, \cL}(X)\geq -C_{\delta, \rho, \cL}$.

By the compactness of $[\delta, 1]$, we assume that $p_r$ convergent to $p_0$.
Let $\widetilde{X}_r=D_\rho^{\frac{1}{2p_r}}X_r D_\rho^{\frac{1}{2p_r}}$ and $P$ is the range projection $X_0$.
We have that $P\neq 0, I$.
Irreducibility of the classical part implies that there exists index $j$ such that $(I-P)B_{p_0, j}P\neq 0$.
After choosing convergent spectral bases along a further subsequence, at least one
coefficient connecting a positive limiting eigenvalue $\mu_\ell \geq c>0$ for some fixed $c$ to a vanishing eigenvalue $\mu_k$ stays nonzero.
The corresponding term in Equation \eqref{eq:energyexpansion} diverges, since
\begin{align*}
\frac{(\mu_{\ell}/\mu_{k})^{1-p_r}-1}{1-p_r} \geq \log \frac{\mu_{\ell}}{\mu_{k}}\to \infty.
\end{align*}
This completes the proof of the lemma.
\end{proof}

% \begin{lemma}
%  there exists a constant $c_{\delta, \rho}$ such that 
%  \begin{align*}
%      \mathcal{E}_{p, \cL_\cN}(X) \geq c_{\delta, \rho} \|X-\bE_{\cN}(X)\|_{HS}^2
%  \end{align*}
% \end{lemma}

\begin{lemma}\label{lem:lowerbd}
Let $0<\eta< \alpha_{\cl}$ and 
\begin{align*}
    U_{\epsilon}=\{X>0: \rho(X)=1, \|X-I\|_{HS}<\epsilon\}.
\end{align*}
For $\epsilon$ small enough, we have that there exists $\kappa_0>0$ such that for all $X\in U_\epsilon$,
 \begin{align*}
     \mathcal{E}_{p, \cL}(X) -\eta \Ent_{p, \rho}(X) +\kappa_0 \mathcal{E}_{p, \cL_{\cN}}(X) \geq \frac{(\alpha_{\cl}-\eta)\delta^2 s_-}{8}\|I-X\|_{HS}^2.
 \end{align*}
\end{lemma}
\begin{proof}
Suppose that $X=I+Y$ with $Y=Y^*$, $\rho(Y)=0$.
Let $Y=Y_1+Y_2$, where $Y_1=\bE_{\cN}(Y)$, $Y_2=Y-\bE_{\cN}(Y)$.
By Equations \eqref{eq:energe1} and \eqref{eq:entropyepsilon}, we have that 
\begin{align*}
 & \mathcal{E}_{p,\cL}(X)-\eta \Ent_{p, \rho}(X) \\
 =& \frac{p^2}{4}\Re \Tr(Y  (\bfW_{p, \rho} \cL  (Y))- \eta\frac{p^2}{2} \Tr( Y \bfR_{p, \rho} Y) +O\left(\|Y\|_{HS}^3\right).
\end{align*}
By taking 
\begin{align*}
    \mathcal{Q}(X_1, X_2)=\frac{p^2}{4}\Re \Tr(X_1^*  \bfW_{p, \rho} \cL  (X_2))- \eta\frac{p^2}{2} \Tr( X_1^* \bfR_{p, \rho} X_2) , \quad X_1, X_2\in M_n(\bC),
\end{align*}
we have that $|\mathcal{Q}(X_1, X_2)|\leq \|\mathcal{Q}\|\|X_1\|_{HS}\|X_2\|_{HS}$, where $\|\mathcal{Q}\|$ is the norm of $\mathcal{Q}$.
Now we have that 
\begin{align*}
\mathcal{Q}(Y, Y)= & \mathcal{Q}(Y_1, Y_1)+ \mathcal{Q}(Y_1, Y_2)+ \mathcal{Q}(Y_2, Y_1)+ \mathcal{Q}(Y_2, Y_2) \\
\geq & \mathcal{Q}(Y_1, Y_1)-\|\mathcal{Q}\|\|Y_1\|_{HS}\|Y_2\|_{HS} -\|\mathcal{Q}\|\|Y_2\|_{HS} ^2.
\end{align*}
Now we shall estimate $\mathcal{Q}(Y_1, Y_1)$ as follows.
Note that 
\begin{align*}
   \mathcal{E}_{p,\cL}(I+\varepsilon Y_1)-\eta \Ent_{p, \rho}(I+\varepsilon Y_1)  \geq (\alpha_{\cl}-\eta)\Ent_{p, \rho}(I+tY_1).
\end{align*}
By Equations \eqref{eq:energe1} and \eqref{eq:entropyepsilon}, we can divide by $\varepsilon^2$ and let $\varepsilon\to 0$ to obtain that 
\begin{align*}
    \mathcal{Q}(Y_1, Y_1)\geq \frac{(\alpha_{\cl}-\eta)p^2 s_-}{2}\|Y_1\|_{HS}^2
    \geq \frac{(\alpha_{\cl}-\eta)\delta^2 s_-}{2}\|Y_1\|_{HS}^2.
\end{align*}
Let $\displaystyle \eta_0= \frac{(\alpha_{\cl}-\eta)\delta^2 s_-}{2}$.
Applying Young's inequality, we obtain that 
\begin{align*}
   2\|\mathcal{Q}\|\|Y_1\|_{HS}\|Y_2\|_{HS} \leq \frac{\eta_0}{2}\|Y_1\|_{HS}^2 +\frac{2\|\mathcal{Q}\|^2}{\eta_0}\|Y_2\|_{HS}^2 .
\end{align*}
Hence
\begin{align*}
\mathcal{Q}(Y, Y)\geq \frac{\eta_0}{2}\|Y_1\|_{HS}^2-\left(\frac{2\|\mathcal{Q}\|^2}{\eta_0} + \|\mathcal{Q}\|\right)\|Y_2\|_{HS}^2.
\end{align*}
On the other hand, we have that
\begin{align*}
    \frac{p^2}{4}\Re \Tr(Y  \bfW_{p, \rho} \cL_\cN  (Y))\geq \frac{p^2s_-}{4}\|Y_2\|_{HS}^2 \geq \frac{\delta^2s_-}{4}\|Y_2\|_{HS}^2.
\end{align*}
By taking 
\begin{align*}
    \kappa_0\geq \frac{4}{\delta^2 s_-} \left(\frac{2\|\mathcal{Q}\|^2}{\eta_0} + \|\mathcal{Q}\|+\frac{\eta_0}{2}\right),
\end{align*}
we have that 
\begin{align*}
\mathcal{Q}(Y, Y)+\kappa_0    \frac{p^2}{4}\Re \Tr(Y  \bfW_{p, \rho} \cL_\cN  (Y))
 \geq \frac{\eta_0}{2}\|Y\|_{HS}^2.
\end{align*}
By choosing $Y$ with $\|Y\|_{HS}$ small enough, we have that 
\begin{align*}
  \mathcal{E}_{p, \rho}(X) -\eta \Ent_{p, \rho}(X) +\kappa_0 \mathcal{E}_{p, \cL_{\cN}}(X) \geq \frac{\eta_0}{4}\|I-X\|_{HS}^2.
  \end{align*}
  This completes the proof of the lemma.
\end{proof}

\begin{lemma}\label{lem:34}
We have that
\begin{align*}
    \mathcal{E}_{p, \cL_{\cN}}(X)\geq \frac{\delta^2}{4} s_-^{2/\delta}\|X-\bE_{\cN}(X)\|_{HS}^2. 
\end{align*}
\end{lemma}
\begin{proof}
Suppose that $\mu_1, \ldots, \mu_n$ are eigenvalues of $X_p$.
Then 
\begin{align*}
    \mathcal{E}_{p, \cL_{\cN}}(X)=& \frac{p^2}{8}\sum_{j,k, \ell} |(P_j)_{k, \ell}|^2 \frac{(\mu_k -\mu_\ell)(\mu_\ell^{p-1} -\mu_k^{p-1})}{1-p}\\
   \geq & \frac{p^2}{8}\sum_{j,k, \ell} |(P_j)_{k, \ell}|^2  (\mu_k -\mu_\ell)^2\\
   \geq & \frac{\delta^2}{8}  \sum_{j=1}^n \|[P_j, X_p]\|_{HS}^2\\
   =& \frac{\delta^2}{4} \|X_p-\bE_{\cN}(X_p)\|_{HS}^2\\
   \geq & \frac{\delta^2}{4} s_-^{2/p}\|X-\bE_{\cN}(X)\|_{HS}^2\\
    \geq & \frac{\delta^2}{4} s_-^{2/\delta}\|X-\bE_{\cN}(X)\|_{HS}^2.
\end{align*}
This completes the proof of the lemma.
\end{proof}

\begin{theorem}\label{thm:pertubation}
Suppose that $\{\Phi_t\}_{t\geq 0}$ is a primitive KMS-symmetric quantum Markov semigroup.
 Fix $0< \delta<1$ and $0<\eta< \alpha_{\cl}$.
 There exists $\kappa_{\delta}(\eta)<\infty$ such that
 \begin{align*}
      \alpha_p(\cL+\kappa\cL_{\cN})\geq \eta,
 \end{align*}
 for all $p\in [\delta, 1]$ and $\kappa \geq \kappa_{\delta}(\eta)$.
 Consequently, we have that 
 \begin{align*}
     \left\|e^{-t(\cL+\kappa \cL_{\cN})}(X)\right\|_{p, \rho} \geq \|X\|_{q, \rho}, 
 \end{align*}
 for all $X>0$,  $\delta < p\leq q <1$ and $\displaystyle t\geq \frac{1}{4\eta}\log \frac{1-p}{1-q}$.
\end{theorem}
\begin{proof}
We have to show that there exists $\kappa>0$ such that for all $p\in [\delta, 1]$ and $X>0$ with
\begin{align}\label{eq:temp}
    \mathcal{E}_{p, \cL}(X) -\eta \Ent_{p, \rho}(X)+ \kappa \mathcal{E}_{p, \cL_{\cN}} (X)\geq 0.
\end{align}
Suppose that we can not find $\kappa>0$ such that the above inequality \eqref{eq:temp} is true.
Then there exists a sequence $\kappa_n\to \infty$, $p_n\in [\delta, 1]$ and $X_n>0$ with $\rho(X_n)=1$ such that 
\begin{align*}
     \mathcal{E}_{p_n, \cL}(X_n) -\eta \Ent_{p_n, \rho}(X)+ \kappa_n \mathcal{E}_{p_n, \cL_{\cN}} (X)< 0.   
\end{align*}
Note that $ \mathcal{E}_{p_n, \cL}(X_n) -\eta \Ent_{p_n, \rho}(X)$ is bounded.
We have that $\mathcal{E}_{p_n, \cL_{\cN}} (X)$ goes to $0$ as $n$ goes to infinity.
By the compactness of $[\delta, 1]$ and $\{X>0: \rho(X)=1\}$, we can choose a convergent subsequence of $\{p_n\}_{n}$.
We can assume that
\begin{align*}
    \lim_{n \to\infty} p_n=p_0, \quad \lim_{n \to \infty} X_n=X_0.
\end{align*}
Moreover, by Lemma \ref{lem:34}, we have that $\bE_{\cN}(X_0)=X_0$.

Now if $X_0>0$ and $X_0\neq I$, then 
\begin{align*}
    \lim_{n \to \infty}\mathcal{E}_{p_n, \cL}(X_n) -\eta \Ent_{p_n, \rho}(X)
    =& \mathcal{E}_{p_0, \cL}(X_0) -\eta \Ent_{p_0, \rho}(X_0) \\
    \geq & (\alpha_{\cl}-\eta)\Ent_{p_0, \rho}(X_0)>0.
\end{align*}
This leads contradiction.
If $X_0$ is singular, by Lemma \ref{lem:infty}, we see that the assumption does not hold.
If $X_0=I$, by Lemma \ref{lem:lowerbd} we see that the assumption does not hold.
Finally, we see that there exists $\kappa>0$ such that Equation \eqref{eq:temp} is true.
% For the diagonal part, we have that 
% \begin{align*}
%      \mathcal{E}_{p, \cL}(I+t Y_1) \geq \alpha_{\cl} \Ent_{p, \rho}(I+tY_1).
% \end{align*}

By the proof of Theorem 11 in \cite{BDR20}, we see that the reverse hypercontractivity follows.
\end{proof}

\section{Graph-based KMS-symmetric Quantum Markov Semigroups}

In this section, we shall study the quantum reverse hypercontractivity for graph-based KMS-symmetric quantum Markov semigroups.

Suppose that $\cG=(V(\cG), E(\cG))$ with the vertex set $V(\cG)=\{1,2, \ldots, n\}$ is a finite undirected graph with a positive weight $\gamma_{jk}>0$ on the edge $e_{jk}$ in $E(\cG)$, where $\gamma_{jk}=\gamma_{kj}$. 
	
	Let $\{E_{jk}\}_{j,k=1}^n$ be the system of matrix units of $M_n(\bC)$.
	Let $\rho$ be a faithful state on $M_n(\bC)$ with density $\displaystyle D_\rho=\sum_{j=1}^n s_j E_{jj}$.
We set $\gamma_{jk}=0$ whenever $e_{jk}\notin E(\cG)$.
Throughout this section, $n\geq2$. 
% The quantity $\lambda_{\cG}$ denotes the second eigenvalue of the
% unweighted graph Laplacian, counting multiplicity, and is therefore
% zero for a disconnected graph.
	
	Let $d_j$ be the weighted degree with respect to $\rho$ at vertex $j$, which is defined as follows:
	\begin{align*}
		d_j=\sum_{k: e_{jk}\in E(\cG)}\gamma_{jk} s_k, \quad d_{\max}=\max_j\{d_j\}.
	\end{align*}
% 	Let $\vec{s}=(s_1, \ldots, s_n)^{\mathsf{T}}$ and
% 	\begin{align}\label{eq:spectralgap}
% 		\varphi_{\cG}=\inf_{\vec{s}\perp \vec{t},\  \vec{t}\neq 0}\frac{\displaystyle \sum_{e_{jk}\in E(\cG)}\gamma_{jk}s_js_k |t_j-t_k|^2}{\displaystyle \sum_{j=1}^n s_j|t_j|^2},
% 	\end{align}
% 	where $\vec{t}=(t_1, \ldots, t_n)^{\mathsf{T}}$.
% 	If $E(\cG)\neq\varnothing$, we have 
% 	\begin{align*}
% 		\frac{(\min_{e_{jk}}\gamma_{jk})\min_{e_{jk}}s_js_k}{\max_j s_j}   \lambda_{\cG}\leq \varphi_{\cG}\leq 2d_{\max},
% 	\end{align*}
% 	For an edgeless graph, $\varphi_{\cG}=d_{\max}=\lambda_{\cG}=0$,
% and the generator defined below is zero.

	For each edge $e_{jk}\in E(\cG)$, we set
	\begin{align}
		V_{jk}=\sqrt{s_j}E_{jk}+\sqrt{s_k} E_{kj}\in M_n(\bC).
	\end{align}
	The Lindbladian associated to the graph $\cG$ is given by
	\begin{align}\label{eq:mnlind1}
		\cL(X) =\sum_{e_{jk}\in E(\cG)} \gamma_{jk} \left(\frac{1}{2}\{V_{jk}^*V_{jk}, X\}- V_{jk}^* X V_{jk}\right), \quad X\in M_n(\bC).
	\end{align}
In \cite{LWW26}, Liu, Wan and the author show that the corresponding semigroup is a primitive KMS-symmetric quantum Markov semigroup, denoted by $\{\Phi_t\}_{t\geq 0}$ whenever $s_1\neq s_2$ for $n=2$ and $\cG$ is connected for $n \geq 3$.
By a direct computation, we have that 
	\begin{align}
		\cL(E_{jj}) =& \sum_{e_{jk}\in E(\cG)} \gamma_{jk} (s_k E_{jj}-s_jE_{kk}), \label{eq:lindbase1}\\
		\cL(E_{jk})= & \frac{d_j+d_k}{2} E_{jk}-\mathbbm{1}_{e_{jk}\in E(\cG)} \gamma_{jk}\sqrt{s_js_k}E_{kj}, \quad j\neq k. \label{eq:lindbase2}
	\end{align}

\begin{lemma}\label{lem:im}
 Suppose that $X>0$ in $M_2(\bC)$, $0<p<1$ and $0<s_2\leq s_1$.
 We have that 
 \begin{align*}
     \Tr(X^{p-1} B X B^*)\geq \Tr(C X^p),
 \end{align*}
 where 
 \begin{align*}
     B=\begin{pmatrix}
         0 & \sqrt{s_1}(s_2/s_1)^{\frac{1}{p}}\\
       \sqrt{s_2}(s_1/s_2)^{\frac{1}{p}}  & 0   
     \end{pmatrix}, \quad C=\begin{pmatrix}
         s_1 & 0\\
      0 & s_2  
     \end{pmatrix},.
 \end{align*}
\end{lemma}
\begin{proof}
Let
\begin{align}
f_p(X)= \Tr(X^{p-1} B X B^*)-\Tr(D_\rho X^p).   
\end{align}
Note that $X$ is a positive definite in $M_2(\bC)$.
There exists $x_1, x_2\in \bR$ such that $X^{p-1}=x_1 X+x_2 I$, where $x_1<0$.
By homogeneity, we assume the eigenvalues of $X$ are $1$ and $\mu$ with $0< \mu <1$.
Let 
\begin{align*}
    b_1= \sqrt{s_2}(s_1/s_2)^{\frac{1}{p}} , \quad b_2=  \sqrt{s_1}(s_2/s_1)^{\frac{1}{p}}, \quad r= (s_2/s_1)^{\frac{1}{p}}.
\end{align*}
If the off-diagonal entry of $X$ is $z$, then the term depending on the phase of $z$ is $2x_1 b_1b_2 \Re (z^2)$.
Note that $x_1<0$.
We can replace $z$ by $|z|$ which only decrease $f_p(X)$.
Hence we assume that the entries of $X$ are real.
Suppose that $g_p(x)=(x-1)(1-x^{p-1})\geq 0$.
Let $\delta\in [0,1]$ be the square first coordinate of a unit eigenvector of $X$ for eigenvalue $1$, i.e.
\begin{align*}
    X=\mu I+(1-\mu)\begin{pmatrix}
    \delta & \sqrt{\delta(1-\delta)}\\ \sqrt{\delta(1-\delta)} & 1-\delta
    \end{pmatrix}
\end{align*}
We have that 
\begin{align*}
    f_p(X)=(1-\delta)s_1g_p(r\mu)+\delta s_2 g_p(\mu/r)-\delta(1-\delta)(b_1+b_2)^2g_p(t).
\end{align*}
We shall show that 
\begin{align}\label{eq:important}
    (b_1+b_2)^2g_p(t)\leq \left(\sqrt{s_1g_p(r\mu)}+\sqrt{s_2 g_p(\mu/r)}\right)^2.
\end{align}
If Equation \eqref{eq:important} is true, then
\begin{align}
    f_p(X) \geq \left((1-\delta)\sqrt{s_1g_p(r\mu)}-\delta \sqrt{s_2 g_p(\mu/r)}\right)^2\geq 0.
\end{align}
We write $\mu=e^{-2x}$ and $r=e^{-2y}$ with $x, y\geq 0$ and 
\begin{align*}
   \widetilde{g}_p(x)=\sqrt{\sinh{x}\sinh{(1-p)x}} .
\end{align*}
Then $\sqrt{g_p(e^{2x})}=2e^{px/2}\widetilde{g}_p(x)$.
Rewrite Equation \eqref{eq:important}, we have that
\begin{align}\label{eq:important1}
    \widetilde{g}_p(x+y)+\widetilde{g}_p(x-y)\geq 2\cosh{\left(\frac{2-p}{2}\right)} \widetilde{g}_p(x).
\end{align}
By taking the second derivative of $\widetilde{g}_p$, we obtain that 
\begin{align*}
    \frac{\widetilde{g}_p''(x)}{\widetilde{g}_p(x)}-\frac{(2-p)^2}{4}=\frac{1}{4}\left( p^2-(\coth{x}-(1-p)\coth((1-p)x))^2\right).
\end{align*}
By the mean-value theorem, we see that $0\leq \coth{x}-(1-p)\coth((1-p)x)\leq p$.
Hence $\displaystyle \widetilde{g}_p''(x)\geq \frac{(2-p)^2}{4}\widetilde{g}_p(x)$.
Now we see that 
\begin{align*}
   \widetilde{g}_p''(x) \geq \widetilde{g}_p(0)\cosh\left(\frac{2-p}{2}x\right)+ \frac{2\widetilde{g}_p'(0)}{2-p}\sinh\left(\frac{2-p}{2}x\right).
\end{align*}
By considering the function $y\mapsto \widetilde{g}_p(x+y)+ \widetilde{g}_p(x-y)$ and noting that its initial value is $2\widetilde{g}_p(x)$ and its initial derivative is $0$, we see that Inequality \eqref{eq:important1} is true.
\end{proof}

\begin{lemma}\label{lem:energypositive}
Suppose that $\{\Phi_t\}_{t\geq 0}$ is a primitive KMS-symmetric quantum Markov semigroup with generator \eqref{eq:mnlind1} and $n \geq 3$.
We have that 
\begin{align*}
    \mathcal{E}_{p, \cL}(X)>0
\end{align*}
for all $X\notin \bC I$.
\end{lemma}
\begin{proof}
Suppose that $X_p= D_\rho^{\frac{1}{2p}} X D_\rho^{\frac{1}{2p}}$,  $B_{p, e}=D_\rho^{\frac{1}{2p}} V_e^* D_\rho^{-\frac{1}{2p}}$ and 
\begin{align*}
\widetilde{X}_{p, e}=(E_{jj}+E_{kk}) X_{p}(E_{jj}+E_{kk}) .
\end{align*}
We have that
\begin{align*}
\frac{4(1-p)}{p^2}\mathcal{E}_{p, \cL}(X)
=& \sum_{e\in E(\cG)}\gamma_e\Tr ( X_p^{p-1} B_{p, e} X_p B_{p, e}^*)-\gamma_e\Tr(V_e^*V_e X_p^p) \\
\geq & \sum_{e\in E(\cG)}\gamma_e\Tr ( \widetilde{X}_{p, e}^{p-1} B_{p, e} \widetilde{X}_{p, e} B_{p, e}^*)-\gamma_e\Tr(V_e^*V_e \widetilde{X}_{p,e}^p)\\
\geq & 0,
\end{align*}
where the first inequality follows from the operator concavity of $t^p$, operator convexity of $t^{1-p}$ and the noncommutative Jensen inequality; the last inequality follows from Lemma \ref{lem:im}.
Hence $\mathcal{E}_{p, \cL}(X)=0$ implies that $X_p$ commutes with $E_{jj}+E_{kk}$ for all edge $\{j,k\}$.
This shows that $X$ is a diagonal matrix.
By taking $X=\diag(x_1, \ldots, x_n)$, we have that 
\begin{align*}
  \mathcal{E}_{p, \cL}(X)=\frac{p^2}{4(1-p)} \sum_{\{j,k\}\in E(\cG)} \gamma_{jk} s_js_k (x_j-x_k)(x_k^{p-1}-x_j^{p-1}).  
\end{align*}
The connectedness implies that $\mathcal{E}_{p, \cL}(X)=0$ if and only if $X$ is a multiple of the identity.
\end{proof}

\begin{lemma}\label{lem:us}
 For a connected graph $\cG$ with $n \geq 3$, the faimily $\{B_{p, e}: e\in E(\cG)\}$ has no common nonzero proper invariant subspace for every fixed $p>0$.   
\end{lemma}
\begin{proof}
 Let $S_e=(s_js_k)^{1/4}(E_{jk}+E_{kj})$.
 Then $B_{p,e}=D_\rho^{\frac{1}{2p}-1/4} S_e D_\rho^{-\frac{1}{2p}+1/4}$ and $S_e$ is self-adjoint.
 Suppose that there is a common nonzero proper invariant subspace whose orthogonal projection is $Q$.
 Then $Q$ commutes with $S_e$ for all $e\in E(\cG)$.
 By the connectedness of $\cG$, we see that $Q$ is $0, I$.
 Hence there is no common nonzero proper invariant subspace for the family.
\end{proof}

\begin{theorem}\label{thm:graphlsi}
Suppose that $n \geq 3$ and the graph $\cG$ is connected.
Suppose that $\{\Phi_t\}_{t \geq 0}$ is a graph-based KMS symmetric quantum Markov semigroup and $0<\delta<1$ is fixed.
Then 
\begin{align*}
    \beta_{\rev, \varepsilon}=\inf_{\delta< p< 1}\alpha_p(\cL)>0.
\end{align*}
\end{theorem}
\begin{proof}
We shall follow the proof from Lemma \ref{lem:infty}.
Suppose that $X_r\to X_0\geq 0$, $\rho(X_r)=1$ with $X_0$ singular, not necessary diagonal.
Let $P$ be the range projection of $D_\rho^{\frac{1}{2p_0}}X_0 D_\rho^{\frac{1}{2p_0}}$.
By Lemma \ref{lem:us}, we have that there exists $e\in E(\cG)$ such that $(I-P)B_{p_0, e}P\neq 0$.
Now by the argument in Lemma \ref{lem:infty}, we see $\mathcal{E}_{p_r,\cL}(X_r)\to \infty$.
Note that $\Ent_{p, \rho}(X)$ is always bounded by $\log s_-^{-1}$ for $\|X\|=1$.
In this case, the quotient $\displaystyle \frac{\mathcal{E}_{p,\cL}(X)}{\Ent_{p, \rho}(X)}$ is bounded arway from zero.

If $X$ is close to $I$, by Equation \eqref{eq:energe1}, we have that 
\begin{align*}
    \mathcal{E}_{p, \cL}(I+tY) 
    = & \frac{p^2 t^2}{4}\Re \Tr(Y\bfW_{p, \rho} \cL(Y))+O(t^3) \\
    \geq &  \lambda_{\KMS} \frac{p^2 t^2}{4}\Re \Tr(Y\bfW_{p, \rho} (Y))+O(t^3),
\end{align*}
where $\lambda_{\KMS}$ is the spectral gap of $\cL$ and the inequality is obtained by the block decomposition of $\cL$ .
This indicates that there exists constant $\beta_\varepsilon$ such that $\mathcal{E}_{p, \cL}(X) \geq \beta_{\varepsilon}\|X-I\|_{HS}^2$.
By Equation \eqref{eq:entropyepsilon}, we see that there exists constant $\widetilde{\beta}_\varepsilon$ such that 
\begin{align*}
    \Ent_{p, \rho}(X)\leq \widetilde{\beta}_{\varepsilon} \|X-I\|_{HS}^2.
\end{align*}
In this case, the quotient $\displaystyle \frac{\mathcal{E}_{p,\cL}(X)}{\Ent_{p, \rho}(X)}$ does not go to $0$ neither.
The remaining part by removing the neighborhood of $I$ and the singular boundary from $\{X>0:\rho(X)=1\}$ is compact.
By the continuity and strict positivity, we see that $\displaystyle\inf_{\varepsilon < p< 1} \alpha_p(\cL) >0$.
This completes the proof of the theorem.
% Suppose that $X>0$ is closed to a singular matrix $X_0$. 
% In the following, we shall show that $\mathcal{E}_{p, \cL}(X)$ is large enough.
% By the fact that $\Ent_{p, \rho}(X)\leq \log s_-^{-1}$, we shall obtain that the infima can not go to zero.

% Let $W_e=\sqrt{\gamma_e}V_e$
% $X_p= D_\rho^{\frac{1}{2p}} X D_\rho^{\frac{1}{2p}}$  $B_{p, e}=D_\rho^{\frac{1}{2p}} W_e^* D_\rho^{-\frac{1}{2p}}$.
% Suppose that $X_p$ is in the diagonal form with eigenvalues $\mu_k$.
% We have that 
% \begin{align*}
%     \frac{4}{p^2} \mathcal{E}_{p, \cL}(X)
% =\sum_{e, k, \ell} |(B_{p, e})_{k\ell}|^2 \mu_\ell^p \frac{(\mu_\ell/\mu_k)^{1-p}-1}{1-p}
% +\frac{\Tr((\sum_e B_{p, e}^* B_{p,e}-\sum_{e}W_e^*W_e)X_p^p)}{1-p}.
% \end{align*}
% By the fact that $X$ is closed to a singular matrix $X_0$, we see that there exists $j_0$ such that $\mu_{j_0}$ close to $0$.
% Note that the norm of $X_p$ is nonzero. 
% There exists $j_1$ such that $\mu_{j_1}$ is its norm.
% Hence 
% \begin{align*}
% \frac{(\mu_{j_1}/\mu_{j_0})^{1-p}-1}{1-p} \geq \log \frac{\mu_{j_1}}{\mu_{j_0}}
% \end{align*}
% is large enough.

\end{proof}

\begin{remark}
Note that for a graph-based primitive KMS-symmetric quantum Markov semigroup without Hamiltonian, the $\displaystyle \inf_{0< p<1}\alpha_p(\cL)\geq 0$. 
However, there exists a graph-based primitive KMS-symmetric quantum Markov semigroup with nonzero Hamiltonian (Example \ref{example:infty}) such that $\displaystyle \inf_{0< p<1}\alpha_p(\cL)=-\infty$.
\end{remark}

\begin{theorem}
Suppose that $\{\Phi_t\}_{t\geq 0}$ is a graph-based KMS-symmetric quantum Markov semigroup with generator \eqref{eq:mnlind1}.
Suppose that $n \geq 3$ and the graph $\cG$ is connected.
    We have that the semigroup has the reverse hypercontractivity \eqref{eq:revhyper} holds for $\varepsilon<p\leq q <1$ whenever
    \begin{align*}
        t\geq \frac{1}{4\beta_{\rev, \varepsilon}}\log \frac{1-p}{1-q}.
    \end{align*}
    Moreover the reverse hypercontractivity holds for $0< p \leq q <1$ whenever
    \begin{align*}
        t\geq \frac{1}{4}\int_p^q \frac{ds}{(1-s)\alpha_s(\cL)}.
    \end{align*}
\end{theorem}
\begin{proof}
    It follows from Theorem 11 in \cite{BDR20} and Theorem \ref{thm:graphlsi}.
\end{proof}

\begin{proposition}
Suppose $n=2$ and the generator $\cL$ is given by
\begin{align*}
    \cL(X)=\gamma\left(\frac{1}{2}( V^*V X+X V^*V)-V^*XV\right),
\end{align*}
where $V=\sqrt{s_1}E_{12}+\sqrt{s_2}E_{21}$, $0<s_2< s_1$.
We have that 
\begin{align*}
  0\leq \alpha_p(\cL) \leq \frac{\gamma p^2}{4(1-p)} \, \frac{(\sqrt{s_1}-\sqrt{s_2})(s_2 (s_1/s_2)^{1/p}-s_1)}{(-\sqrt{s_1}\log s_2 -\sqrt{s_2} (s_1/s_2)^{1/p}\log s_1 )}.
\end{align*} 
Consequently, $\displaystyle \inf_{0<p<1} \alpha_p(\cL) =0$.
\end{proposition}
\begin{proof}
Let $\displaystyle \vec{u}=\frac{1}{\sqrt{\sqrt{s_1}+\sqrt{s_2}}} \begin{pmatrix}
    s_2^{1/4} \\ s_1^{1/4}
\end{pmatrix}$ be an eigenvector of $V^*$.
Then $V^*\vec{u}=(s_1s_2)^{1/4}\vec{u}$.
Let $P$ be the rank-one projection corresponding to $\vec{u}$.
Suppose that $X_\varepsilon=P+\varepsilon I$ with $\varepsilon>0$.
We have that
\begin{align*}
  \Ent_{p, \rho}(X_\varepsilon)
   \to & -\|D_\rho^{\frac{1}{2p}} P D_\rho^{\frac{1}{2p}}\|_{HS}^p\Tr\left(\frac{D_\rho^{\frac{1}{2p}} P D_\rho^{\frac{1}{2p}}}{\|D_\rho^{\frac{1}{2p}} P D_\rho^{\frac{1}{2p}}\|_{HS}^2}\log D_\rho\right), \\
   \mathcal{E}_{p, \cL}(X_\varepsilon) \to & \frac{\gamma p^2 \|D_\rho^{\frac{1}{2p}} P D_\rho^{\frac{1}{2p}}\|_{HS}^p }{4(1-p)}\left(\sqrt{s_1s_2}-\Tr\left(\frac{D_\rho^{\frac{1}{2p}} P D_\rho^{\frac{1}{2p}}}{\|D_\rho^{\frac{1}{2p}} P D_\rho^{\frac{1}{2p}}\|_{HS}^2}V^*V\right)\right)
\end{align*}
as $\varepsilon\to 0$, we see that the proposition is true.
\end{proof}

\bibliographystyle{abbrv}
\bibliography{kms}
\end{document}